\documentclass[11pt]{article}

\usepackage[ruled, vlined, linesnumbered, noresetcount]{algorithm2e}
\usepackage{makeidx}  
\usepackage{amsfonts,amsmath,amssymb,url,amsthm}
\usepackage{epsfig}
\usepackage{graphics}
\usepackage{color}
\usepackage{colordvi}

\DeclareMathOperator*{\argmax}{arg\,max}

\DeclareMathOperator*{\argmin}{arg\,min}

\begin{document}

\baselineskip=0.65truecm

\SetKwFor{ForAll}{forall}{do}{end forall}

\newtheorem{prop}{Proposition}
\newtheorem{theorem}{Theorem}
\newtheorem{corol}{Corollary}
\newtheorem{fact}{Fact}
\newtheorem{lemma}{Lemma}
\newtheorem{definition}{Definition}
\newtheorem{remark}{Remark}
\newtheorem{corollary}{Corollary}

\def\MIS{M}
\def\spec{\infty}

\def\Amax{A}
\def\Bmax{B}
\def\Cmax{B}
\def\Dmax{C}

\newcommand{\GC}[1]{{\color{blue}\textit{\textbf{Gennaro:} #1}}}
\newcommand{\remove}[1]{}

\newcommand{\Evg}{{{\sf Evg}}}
\newcommand{\Inf}{{{\sf Inf}}}
\newcommand{\A}[1]{\Evg[#1]}
\newcommand{\I}[1]{\Inf[#1]}

\newcommand{\Ag}[2]{\Evg_{#1}[#2]}
\newcommand{\Ig}[2]{\Inf_{#1}[#2]}

\newcommand{\N}{{\mathbb{N}}}

\def\ti{t_I}
\def\ta{t_E}

\newcommand{\n}[2]{NO_{#1}[#2]}
\newcommand{\nwI}[2]{\widehat{NO}_{#1}[#2]}
\newcommand{\nn}[2]{Inf_{#1}[#2]}
\newcommand{\nnwI}[2]{\widehat{Inf}_{#1}[#2]}
\newcommand{\nnn}[2]{Evg_{#1}[#2]}
\newcommand{\nnnwI}[2]{\widehat{Evg}_{#1}[#2]}

\def \t {{{t}}}
\def\V{{\cal V}}
\def\bs{{\bf s}}
\def \vc {{{\sc VC}}}
\def\X{{\cal X}}
\def\T{{\cal T}}
\def\taa{T_E}
\def\tii{T_I}
\def\GX{G_X}
\def\GY{G_Y}
\def\GZ{G_Z}
\def\OGX{OPT_{G_X}}
\def\OGXo{OPT_{G_X^0}}
\def\OGXi{OPT_{G_X^i}}
\def\OGXii{OPT_{G_X^{i-1}}}
\def\OGY{OPT_{G_Y}}
\def\OGZ{OPT_{G_Z}}
\def \IA {H}

\def\MSS{{{MES\ }}}

\title{Evangelism in Social Networks: Algorithms and Complexity
\thanks{An extended abstract  of this paper was presented at the 27th International Workshop on Combinatorial Algorithms (IWOCA 2016), Helsinki, Finland, August 17--19, 2016 \cite{CGRV16}}}

\author{
  Gennaro Cordasco\\
	Seconda Universit\`a degli studi di Napoli, Caserta, Italy\\
  \texttt{gennaro.cordasco@unina2.it}
  \and
  Luisa Gargano, Adele Anna Rescigno and Ugo Vaccaro\\
	Universit\`a degli Studi di Salerno, Fisciano, Italy\\
  \texttt{\{lgargano,arescigno,uvaccaro\}@unisa.it}
}

%
%

\maketitle

\begin{abstract}
We consider a population of interconnected  individuals that, with
respect to a piece of information, at each  time instant 
can be   subdivided into three (time-dependent) categories: 
{\em agnostics, influenced}, and {\em evangelists}.
A dynamical process of information diffusion  evolves among the individuals 
of the population  according to the following rules.
Initially, all individuals are agnostic.
Then,  a set of people   is chosen  from the outside and  
convinced to start evangelizing, i.e., to start spreading the information.
When a number of evangelists, greater than a given  threshold, 
communicate  with a  node $v$, the  node $v$ becomes influenced, whereas, 
as soon as the  individual  $v$ is contacted  by a sufficiently \emph{much larger}
number of evangelists,   it is itself converted into  an evangelist and consequently 
it starts spreading the information. The question is: How to choose a bounded cardinality  
initial set of evangelists so as to maximize the final number of \emph{influenced} 
individuals? We prove that the  problem is hard to solve, even  in an  
approximate sense. On the positive side,  we present exact polynomial time algorithms for trees and 
complete graphs.  For general graphs, we derive exact   parameterized algorithms.
We also investigate   the problem  when
the objective is to select a minimum number of evangelists capable of influencing the whole network.
Our motivations to study these problems come from the areas of Viral Marketing and 
the analysis of quantitative models of spreading  of
influence in social networks.
\end{abstract}

\section{The Context}
Customer Evangelism \cite{MCH}
 occurs when a customer
\emph{actively} tries to convince other customers to buy 
or use a particular brand. Fathered by Apple marketing \emph{guru} Guy Kawasaki in the 90's \cite{Kaw}, the idea of consumer evangelism has found
a new and more powerful incarnation 
in modern communications media. Social networks like Twitter, Facebook and Pinterest
have indeed immensely empowered properly
motivated individuals towards brand advocacy and proselytism. 
We plan to 
abstract a few algorithmic problems out of  this scenario,  and provide 
 efficient solutions for some of them.

\section{The Problem}
Our model   posits  an interconnected  population consisting of  individuals that, with
respect to a piece of information and/or an opinion, at each  time instant can be   subdivided into three 
time-dependent categories: 
{\em agnostics, influenced}, and {\em evangelists}.
Initially, all individuals are agnostic.
Then,  a set of people   is chosen  and  converted into  evangelists, that is, convinced to start
spreading the information. When a sufficiently large number of
evangelists
communicate  with an  node $v$, the  node $v$ becomes influenced; as soon as the 
individual  $v$ has in his neighborood a  \emph{much larger} 
number of evangelists,   it is converted to  an evangelist and \emph{only then} it starts spreading the information itself.
Our  model  can  be seen also as an idealization of  
diffusion processes  studied  in the area of {memetics}.
A \emph{meme} \cite{D}  is a convinction, behavior, or fashion  that spreads from person to person within a culture.
It is apparent   that not every  meme learned by a person   spreads among the individuals of a population. We are 
making here the reasonable hypothesis  that individuals indeed  acquire  a meme when it has been heard of   from a few friends,
but people  start spreading the same meme only when they believe it is popular, fashionable, or important, i.e., 
when  it has been communicated to them by a \emph{large} number of friends. This is not too far
from what has been experimentally observed about  how  memes evolve and  spread  within Facebook \cite{Ad+}.

A bit more concretely, we are given 
a graph $G = (V,E)$, abstracting a social network,  where the node set $V$ corresponds to 
people and the edge set to relationships among them. We denote by $N_G(v)$  the neighborhood of node $v\in V$ 
and by $d_G(v)= |N_G(v)|$ the degree  of $v $ in $G$, we avoid the subscript $G$ whenever  the graph is clear from the context.
Moreover,  let $\ti: V \to  \{0,1,2,\ldots\}$  and  $\ta: V \to \{0,1,2,\ldots\}$ be two functions assigning 
integer thresholds to the nodes in  $G$ such that 
  $0\leq \ti(v)\leq\ta(v)\leq d(v)+1$, for each  $v\in V$.

An evangelization   process in $G$, starting at a subset of nodes 
$S \subseteq V$, is characterized by two sequences of node subsets
$$\A{S,0} \subseteq \A{S,1} \subseteq \ldots \subseteq \A{S,\tau} \subseteq \ldots \subseteq V,$$ and 
$$\I{S,0} \subseteq \I{S,1} \subseteq \ldots \subseteq \I{S,\tau} \subseteq \ldots \subseteq V,$$
where for each $\tau=0,1,\ldots,$ it holds that $\A{S,\tau}\subseteq \I{S,\tau} $. The  process  is formally described by the following dynamics:
\begin{align*}
\A{S,0}&=\I{S,0} = S \mbox{, and and for each } \tau{\geq} 1\\
\A{S,\tau} &= \A{S,\tau{-}1}\cup \Big\{u : \big|N(u)\cap \A{S,\tau {-} 1}\big|\ge \ta(u)\Big\},\\
\I{S,\tau} &= \I{S,\tau{-}1} \cup \Big\{u : \big|N(u)\cap \A{S,\tau {-} 1}\big|\ge \ti(u)\Big\}.
\end{align*}
In  words, a node $v$ becomes influenced if  the  number of its evangelist neighbors is greater than or equal to
its influence threshold $\ti(v)$, and  $v$  becomes an evangelist  if the  number of evangelists among  its 
  neighbors reaches
its evangelization  threshold $\ta(v)\geq \ti(v)$.
The  process terminates when $\A{S,\rho}=\A{S,\rho-1}$ for some $\rho\geq 0.$ 
We denote by $\A{S}=\A{S,\rho}$ and $\I{S}=\I{S,\rho}$ the final  sets  when the process terminates.
The initial set $S$ is also denoted as a {\em seed set} of the evangelization process.
Due to  foreseeable difficulties  in  hiring evangelists, it seems reasonable
trying to limit their initial  number,  and see how the dynamics of the spreading process evolves. Therefore, we 
  state our    problems as follows:

\smallskip
\noindent
{\sc  Maximally Evangelizing   Set (\MSS\!\!)}.\\
{\bf Instance:} A graph $G=(V,E)$, thresholds $\ti, \ta:V\to \{0, 1, 2, \ldots \}$,
and  a budget $\beta$.\\
{\bf Question:} Find a seed set $S\subseteq V$, with  $|S| \leq \beta$, such that  $|\I{S}|$ 
is \emph{maximum}.

\smallskip
\noindent
{{\sc  Perfect Evangelizing   Set (PES)}.\\
{\bf Instance:} A graph $G=(V,E)$, thresholds $\ti, \ta:V\to \{0, 1, 2, \ldots \}$.\\
{\bf Question:} Find a seed set $S\subseteq V$ of \textit{minimum size} such that  $\I{S}=V$.}

{	It is worth to mention that the  PES problem is, in a sense, easier than the \MSS\ problem.
Indeed,  any  algorithm that solves the \MSS\ problem can be easily adapted to the PES problem 
by means of a standard binary search argument, while the opposite it is not true.}

\section{What is already known and what {we} prove}
The above algorithmic problems have  roots in the broad area 
of the \emph{spread of influence}
in Social Networks (see  \cite{CLC,EK} and references quoted therein).
In the introduction of this  paper we have already  highlighted the connections of our model to the 
general area of viral marketing.  There, 
companies  wanting to
promote products or behaviors might try initially  to target and convince
a few individuals which, by word-of-mouth effects, can  trigger
a  cascade of influence in the network, leading to
an  adoption  of the products by  a much larger number of individuals.
Not unexpectedly, viral marketing  has also become  an important tool in
 the communication strategies of politicians \cite{LKLG,T}.
Less secular applications of our evangelization process can also be envisioned.
  Here, we shall limit ourselves 
to discuss the work that is most directly related to ours, and refer the reader to the 
 authoritative texts  \cite{CLC,EK} for a synopsis of the area.
The first authors to study  spread of influence in networks
from an algorithmic point of view were Kempe \emph{et al.}, see \cite{KKT-15}.
However, they were mostly interested in networks with  randomly chosen thresholds.
Chen \cite{Chen-09} studied the following minimization problem:
given an unweighted  graph $G$ and fixed thresholds $t(v)$, for each vertex $v$ in $G$,
find
a  set of minimum size that eventually influences
all (or a fixed fraction of) the nodes of $G$.
He proved  a  strong inapproximability result that makes unlikely the existence
of an  algorithm with  approximation factor better than  $O(2^{\log^{1-\epsilon }|V|})$.
Chen's result stimulated a series of papers, e.g.,   
\cite{ABW-10,BCNS,BHLM-11,Centeno12,Chopin-12,Chun,Cic14,Cic+,C-OFKR,CO+,Ga+,LNK,NNUW,Re} that isolated interesting cases 
in which the problem (and variants thereof) become tractable.

All  of the above quoted  papers considered the basic model in which \emph{any} node,
as soon as it is  influenced by its neighbors,  immediately starts spreading influence. 
The more refined model put forward in this paper, that differentiate among active 
spreaders (evangelists) and plain informed (influenced)  nodes, appears to be new, to the best of our knowledge.
{We would like to point out that we obtain  an interesting information diffusion model
already in the particular case in which $\ti(v)=1$, for each node $v$. 
In fact, in this case nodes in the sets
$\I{S,\tau}$  would correspond to people that have simply  \emph{heard} about a piece of information, while people 
in the sets  $\A{S,\tau}$ would correspond  to people who are \emph{actively}
spreading  that same piece of  information.}


{In Section \ref{sec:complexity}, we {first} prove that the  MES  problem is hard to solve, even  in the 
approximate sense. {Subsequently, we design} exact algorithms, for the MES problem,  parameterized with respect
to neighborhood diversity (and, as a byproduct, by  vertex cover) and for the PES problem parameterized with respect
to the treewidth.
In Section \ref{sec-trees}, we present exact polynomial time algorithms  for the MES problem on complete graphs  and trees. 
Finally, in Section \ref{sec-dense} we study the PES problem in dense graphs.} 

\section{MES  is hard, also to approximate}\label{sec:complexity}
The \MSS problem includes   the  \textsc{Influence Maximization} (IM) problem \cite{KKT-15}, 
that  is  known to be NP-hard to approximate within a ratio of $n^{1-\epsilon}$,  for any 
$\epsilon > 0$. 
In our terminology, the IM problem  takes in input 
 a graph  $G$ with a threshold function   $t:V\to \{0,1,2,\ldots\}$
and a budget $\beta$, 
and asks for a subset $S$ of  $\beta$ nodes of $G$ such that $|\A{S}|$ is maximum.
An  instance of the IM problem 
corresponds to the  
 \MSS instance consisting  of  $G$, $\beta$,  and threshold functions $\ta, \ti$,  with  
$\ti(v)=\ta(v)=t(v)$, for each  $v\in V$. Here
we show that  the \MSS problem remains hard even if the influence threshold $\ti$ is equal to
$1$,  for each node $v\in V$.

%

\begin{theorem}\label{teo-compl}
It is NP-hard to approximate the  \MSS problem  within a ratio of $n^{1-\epsilon}$ for any 
$\epsilon > 0$ even when $\ti(v)=1$, for each  node $v \in V$.
\end{theorem}

\begin{proof}
We  construct a gap-preserving reduction from the 
{\sc  Influence Maximization (IM)} problem.
The 
theorem follows from the inapproximability of influence maximization problem proved in \cite{KKT-15}.
Consider an instance of the
IM problem consisting in  a graph $G = (V,E)$ 
with threshold function $t(\cdot)$ and bound $\beta$.
Let $V=\{v_1,\ldots,v_n\}$,
 we build a graph $G' = (V',E')$ having $n(n+1)$ nodes, as follows:
	\begin{itemize}
		\item Replace each $v_i\in V$ by a gadget $G'_i$ consisting in a star  in which the node set is 
		$V'_i=\{v_{i,0},v_{i,1},\ldots, v_{i,n}\}$ and the center $v_{i,0}$ is connected with each of  the other nodes  $v_{i,1},\ldots, v_{i,n}$.
		Formally,	
		
		 $\quad - V'=\bigcup_{i=1}^n V'_i=\{v_{i,j}\ |\ 1\leq i\leq n,\, 0\leq j\leq n\}$ 
		
  $\quad - E'=\{(v_{i,0},v_{\ell,0}) | 1\leq i<\ell\leq n,\, (v_i,v_\ell) \in E\} \bigcup \{ (v_{i,0},v_{i,j}), \mbox{ for } i,j=1,\ldots, n\}$.
		\item the node $v_{i,0}$ has threshold $\ta(v_{i,0})=t(v_i)$, while each other node 
		$v_{i,j}\in V'_i$  with $j\geq 1$ has $\ta(v_{i,j})=1$, for $i=1,\ldots, n$.
	\end{itemize}
 Notice that  $G$ corresponds
to the subgraph of $G'$ induced by the set $\{v_{i,0}  | 1\leq i \leq n \}$.
 Hence, for each star $G'_i$ in $G'$, the center $v_{i,0}$ plays the role
of $v_i$ in $G$. 
Moreover, it is worth mentioning that during an
evangelization process in $G'$ if the  node $v_{i,0}$ in the  gadget $G'_i$ is an evangelist, then all the nodes in $V'_i$
will be influenced  within the next round.

\noindent
We prove that: {\em  There exists a seed set $S \subseteq V$ for $G$  of size $|S| = \beta$ s.t. $|\Ag{G}{S}|\geq k$ iff there exists a seed set $S'\subseteq V'$ for $G'$ of cardinality $|S'|=\beta$ such that $|\Ig{G'}{S'}| \geq k (n+1).$ }

Assume that $S \subseteq V$ is a seed set for $G$ such that $|S| =\beta$ and  $|\Ag{G}{S}|\geq k$, we can easily build a seed set for $G'$ as 
$S' = \{v_{i,0} \in V' | v_i \in S\}$.
Clearly, $|S'| = |S|$. To see that $|\Ig{G'}{S'}| \geq k (n+1)$ we notice that since  $G$ is isomorphic to the subgraph of $G'$ induced by $\{v_{i,0} \in V'_i | v_i \in V \},$ all the nodes $v_{i,0} \in   V'_i$ such that $v_i \in \Ag{G}{S}$  will become evangelists. Then once a node $v_{i,0}$ becomes an evangelist, the  nodes  
$\{v_{i,1},v_{i,2},\ldots,v_{i,n}\}$ will be influenced  in the next round.
 Hence $|\Ig{G'}{S'}| \geq (n+1) \times|\Ag{G}{S}|\geq k (n+1).$

On the other hand, assume that $S'\subseteq V'$ is a seed  set for $G'$ 
such that $|S'| = \beta$ and  $|\Ig{G'}{S'}|\geq k(n+1)$, we can easily build a seed set for $G$ as 
$S = \{v_i \in V | \  S'\cap V'_i\neq\emptyset\}$.
By construction $|S| \leq |S'|.$ 
Let $S''=\{v_{i,0} \in V' \ |\  S'\cap V'_i\neq\emptyset\}$. 
It is easy to observe that $|\Ig{G'}{S''}| = |\Ig{G'}{S'}|\geq k(n+1)$.
Let $V'_{0}=\{v_{i,0}  \ |\  1\leq i\leq n \}$; to see that $|\Ag{G}{S}| \geq k$ we will show that 
$|\Ag{G'}{S''} \cap V'_0| \geq k$. The result will follow since $G$ is isomorphic to the subgraph of $G'$ induced by 
$V'_0.$ 
In order to show that  $|\Ag{G'}{S''} \cap V'_0| \geq k$, we first observe  that  
$|\Ig{G'}{S''} \cap (V'-V'_0) | \geq |\Ig{G'}{S''}  |-|V'_0|\geq k(n+1)-n$. 
Nodes in $V'-V'_0$ can be influenced only by nodes in 
$V'_0$. 
Moreover, a node in $V'_0$ can influence at most $n$ nodes in $V'-V'_0$---the leaves of the star of which it is the center. Hence in order to influence $k(n+1)-n$ nodes in $V'-V'_0$ at least $	\left \lceil \frac{k(n+1)-n}{n} \right \rceil \geq k$ nodes must be evangelist among those in $V'_0$ and consequently  $|\Ag{G'}{S''} \cap V'_0 | \geq k.$ 

\end{proof}

\remove{
\proof {\em (Sketch.)}
We  construct a gap-preserving reduction from the 
{\sc   IM} problem.
The 
theorem follows from the corresponding inapproximability result proved in \cite{KKT-15}.
Consider an instance of the
IM problem consisting in  a graph $G = (V,E)$, 
 threshold function $t(\cdot)$, and bound $\beta$.
Let $V=\{v_1,\ldots,v_n\}$,
 we build a graph $G' = (V',E')$ as follows:

	\begin{itemize}
		\item Replace each $v_i\in V$ by a gadget $G'_i$ consisting in a star  in which the node set is 
		$V'_i=\{v_{i,0},v_{i,1},\ldots, v_{i,n}\}$ and the center $v_{i,0}$ is connected to each of  the other nodes  $v_{i,1},\ldots, v_{i,n}$.
		Formally,	
		
		 $ - V'=\bigcup_{i=1}^n V'_i=\{v_{i,j}\ |\ 1\leq i\leq n,\, 0\leq j\leq n\}$ 
		
  $ - E'=\{(v_{i,0},v_{\ell,0}) | 1\leq i<\ell\leq n,\, (v_i,v_\ell) \in E\} \bigcup \{ (v_{i,0},v_{i,j})|  i, j=1,\ldots, n\}$;
		\item the node $v_{i,0}$ has threshold $\ta(v_{i,0})=t(v_i)$, while each other node 
		$v_{i,j}\in V'_i$  with $j\geq 1$ has $\ta(v_{i,j})=1$, for $i=1,\ldots, n$.
	\end{itemize}

	\noindent
 Notice that  $G$ corresponds
to the subgraph of $G'$ induced by the set $\{v_{i,0}  | 1\leq i \leq n \}$.
 Hence, for each star $G'_i$ in $G'$, the center $v_{i,0}$ plays the role
of $v_i$ in $G$. 
Moreover, it is worth mentioning that during an
evangelization process in $G'$,  if the  node $v_{i,0}$ in the  gadget $G'_i$ is evangelizing, then all the nodes in $V'_i$
will be influenced   within the next round.
The essence of our  proof  of the theorem consists in the following claim, that will be fully justified in  Appendix \ref{App1}:
{\em  There exists a seed set $S \subseteq V$  for $G$ of size $|S| = \beta$ s.t. $|\Ag{G}{S}|\geq k$ 
iff there exists a seed set $S'\subseteq V'$ for $G'$ of cardinality $|S'|=\beta$ such that $|\Ig{G'}{S'}| \geq k (n+1).$ } \qed
}

\section{Parameterized complexity } \label{parametr}
	
A parameterized computational problem with input size $n$ and  parameter $\t$ is called {\em fixed parameter tractable (FPT)} 
 if it can be solved in time $f(\t) \cdot n^c$, where $f$ is a function  depending on $\t$ only,  and $c$ is a constant \cite{DF}.
\remove{
When the problem turns out to be fixed parameter intractable with respect to $\t$ then the problem is said to be $W[P]$-complete.
In case the problem is polynomial solvable but the degree of the polynomial of the running time depends on the parameter $\t$, then the problem is in the complexity class XP when parameterized by $\t$.
}
In this section we study the effect of some parameters on the computational complexity of the \MSS and  PES  problems.

\subsection{Parameterization of \MSS\ with Neighborhood Diversity.}
We consider the decision version ($\alpha,\beta$)-\MSS\!\! of the problem.
   It takes in input a
graph $G=(V,E)$, node thresholds $\ti:V\to \{0,1,2,\ldots\}$ and $\ta:V\to\{0,1,2,\ldots\}$, and integer bounds $\alpha,\beta \in \N$, and  
 asks if there exists  a seed set $S\subseteq V$  such that $|S| \leq \beta$ and $|\I{S}|\geq \alpha$.
\\
We notice that by conveniently choosing the thresholds $\ta$ and  $\ti$, the \MSS problem specializes in problems whose parameterized complexity is well known. 
When $\ti(v) = \ta(v)$ for each $v \in V$ and $\alpha = |V|$, the problem becomes the {\em target set selection}  \cite{Chen-09}. This problem is $W[2]$-hard\footnote{See  \cite{DF} for  definitions of $W[2]$-hardness, $W[1]$-hardness and the class XP.} with respect to the solution size $\beta$ \cite{NNUW}, it is XP when parameterized with respect to the treewidth \cite{BHLM-11}, and 
is $W[1]$-hard with respect to the parameters treewidth, cluster vertex deletion number and pathwidth \cite{BHLM-11,Chopin-12}.
Moreover, the  target set selection problem becomes fixed-parameter tractable with respect to the single parameters: Vertex cover number, feedback edge set size, bandwidth \cite{Chopin-12,NNUW}. 
In general  when $\ti(v) = \ta(v)$ for each $v \in V$,  the ($\alpha,\beta$)-\MSS\!\! problem   has no parameterized approximation algorithm with respect to the parameter $\beta$ and it is $W[1]$-hard with respect to the combined parameters $\alpha$ and 
$\beta$ \cite{BCNS}.
{Moreover, the  target set selection problem is W[1]-hard parameterized
by the neighborhood diversity of the input graph \cite{DKT16}.}

{In the following, we study the parameterized complexity of the ($\alpha,\beta$)-\MSS\!\!   problem for  the general case   $\ti(v)\neq \ta(v)$.  
We concentrate our attention on two parameters: the neighborhood diversity and the vertex cover size.}

The {\em neighborhood diversity} was  first introduced in \cite{L}. It  has recently received  
 particular attention \cite{DKT16,FGKKK,G,GR} also 
due to its  property  of being computable  in polynomial time 
\cite{L}---unlikely other parameters, including treewidth, rankwidth, and vertex cover.

\begin{definition}\label{diversity}
Given a graph $G=(V,E)$,  two nodes $u,v\in V$  have the same  
{\em type} iff $N(v) \setminus \{u\} = N(u) \setminus \{v\}$.
The graph $G$ has {\em neighborhood diversity} $\t$, 
if there exists a partition of $V$ into at most $\t$ sets, $V_1,V_2, \ldots, V_\t$, 
s.t. all the 
nodes in   $V_i$ have the same type, for  $i=1,\ldots, \t$. 
The family 
$\V=\{V_1,V_2, \ldots, V_\t\}$ is called  the {\em type partition} of $G$.
\end{definition}

Let $G=(V,E)$ be a graph with type partition  $\V=\{V_1,V_2, \ldots, V_\t\}$. 
By Definition \ref{diversity}, 
 each $V_i$ induces either a clique or an independent set in $G$.
For each  $V_i,V_j\in \V$, we get that  either each node in 
$V_i$ is a neighbor of each node in $V_j$ or no node in $V_i$ 
has a neighbor  in $V_j$.
Hence, all the nodes  in the same $V_i$   have the same neighborhood  $N(V_i)$---excluding the nodes in $V_i$ itself. 

We present a FPT-algorithm for the 
\MSS problem with parameters $\t$ and $\beta$.
%
At the end of the evangelization  process in $G$ starting at $S$, we identify 
the number of evangelists that are neighbors of (all) the nodes in $V_i$ and  define 
for each $i=1,2,\ldots,t,$

$$N_i(S) = \begin{cases}{
|\A{S} \cap N(V_i)|}&{\mbox {if $V_i$ is an independent set,}}\\
 {|\A{S} \cap (V_i \cup N(V_i))|}&{\mbox {if $V_i$ is a clique.}}\end{cases}$$

\noindent
It is easy to see that a node  $u \in V_i-\A{S}$ is influenced if $\ti(u) \leq N_i(S) $.

The proposed algorithm will be based on the following Lemma.
\begin{lemma}\label{lemma1-n-div}
Let $S'$ be a seed set  for $G$.
Let $u,v \in V_i$ be s.t. {$u \in S'$ and $v \not \in S'$, }and consider the set $S''=(S'-\{u\}) \cup \{v\}$. 
If $\ti(v) > N_i(S')$ then $\I{S'} \subseteq \I{S''}$.
\end{lemma}

\begin{proof}
Consider   a seed set $S'$ for $G$.
For $u,v \in V_i$  such that  $u \in S'$ and $v \not \in S'$
consider $S''=(S'-\{u\}) \cup \{v\}$. 

It is trivial to see that after the first round of the  evangelization process   with seed set  $S''$, the number of influenced nodes (resp.  of evangelists) in each  $V_i \in \V$ is 
the same as with  seed set $S'$. Namely, 
\begin{equation}\label{first-round}   
|N(v) \cap S'| = |N(u) \cap S''| \mbox{ and  $|N(w) \cap S'| = |N(w) \cap S''|$ for each $w \in V-\{u,v\}$}.
\end{equation}
Let $\ti(v) > N_i(S')$.
Note that since $\ta(v) \geq \ti(v) > N_i(S')$, node $v$ does not take part to make any node an evangelist in the evangelization 
process starting at $S'$. 
To prove the lemma we distinguish two cases according to the value of $\ta(u)$.\\
- If $\ta(u) \leq N_i(S')$ then there exists a round $i$ of the process starting at $S'$ in which $u$ becomes an evangelist, that is, $|N(u) \cap \A{S',i-1}| \geq \ta(u)$. 
Consider now the evangelization  process starting at $S''$.
By (\ref{first-round}), the effect on any node of the process starting at $S''$ at the end of the first round is the same of the process starting at $S'$ at the end of the first round.
Furthermore, till round $i-1$ of the process starting at $S''$, the evangelists and the influenced nodes  are exactly the same of the corresponding ones  of the process starting at $S'$. Hence  at round $i$ of the process starting at $S''$, node $u$ becomes an evangelist and $\A{S',i} \subseteq \A{S'',i}$. 
In the following rounds $j > i$ the relation $\A{S',j} \subseteq \A{S'',j}$ is retained, and at the end of the process we have $\A{S'}\subseteq \A{S''}$. Since $\A{S'}\subseteq \A{S''}$ implies $\I{S'} \subseteq \I{S''}$, the lemma is proved in this case.\\
- Let $\ta(u) < N_i(S')$. 
By (\ref{first-round})   and considering that during 
the  process starting at $S''$, the set of evangelists grows exactly as the set of evangelists  in  the   process starting at $S'$ we have that the evangelization  process starting at $S''$ proceeds exactly as the  process starting at $S'$ and at the end of the process it holds $\A{S'}=\A{S''}$ and $\I{S'}=\I{S''}$. 
\end{proof}

We  now   present our algorithm.
We assume that the nodes of $G$ are sorted  in order of non--increasing evangelization thresholds and consider all the possible $\t$-ples $(s_1,s_2,\ldots,s_\t)$ such that $\sum_{i=1}^{\t} s_i =\beta$.
For each   $\bs=(s_1,s_2,\ldots,s_\t)$ we construct the set 
$S_\bs$ in two steps. In the first step we set
$S_\bs= \cup_{i=1}^{\t} S_i$ 
where $S_i$ is obtained by choosing $s_i$ nodes with the largest evangelization threshold  in $V_i$.
In the second step we first consider the evangelization process in $G$ starting at $S_\bs$ and then we update each $S_i$ by using the nodes that have not been influenced in the process. 
In particular, $S_i$ is updated by replacing as many nodes  as possible among those that  could be influenced (if outside $S_i$)
by nodes that 
cannot be 
influenced. 
The construction of $S_\bs$ is detailed in  algorithm ME-ND($\bs,\V$).
We then consider the evangelization process  in $G$ starting at $S_\bs$ and  get the number  $\alpha_\bs= |\I{S_\bs}|$ of  influenced nodes at the end of the process.
Finally, we determine $\bs' = {\argmax}_{\bs} \ \alpha_\bs$ and compare $\alpha$ with $\alpha_{\bs'}$.
If $\alpha_{\bs'} \geq \alpha$ then we answer {\sc yes} to the  \MSS  question for $G$   with parameters 
$\alpha$ and $\beta$  and  $S_{\bs'}$ is the desired  seed set;
otherwise we answer {\sc no}.

  \begin{algorithm}
\SetCommentSty{footnotesize}
\SetKwInput{KwData}{Input}
\SetKwInput{KwResult}{Output}
\DontPrintSemicolon
\caption{ \ \    ME-ND($\bs,G$) \label{alg1}}
\KwData {A graph $G=(V,E)$, threshold functions  $\ti$ and $\ta$ and  $\bs=(s_1,s_2,\ldots,s_\t)$; a type partition of $G$.}
\KwResult{$S_\bs= \cup_{i=1}^{\t} S_i$, a seed set for $G$ such that for each $i=1,2,\ldots,t,$ $s_i=|S_i|$ }
\setcounter{AlgoLine}{0}
\ForEach {$i=1,\ldots,\t$}{Let $S_i$ be a set of  $s_i$ nodes of $V_i$ with the largest evangelization thresholds (e.g., for any $u\in S_i$ and $v\in V_i-S_i$ it holds $\ta(u)\geq \ta(v)$). }
Set $S_\bs = \cup_{i=1}^{\t} S_i$ and  consider the process  in $G$ starting at $S_\bs$.\\
\ForEach (\tcp*[f]{Update set $S_i$;}){$i=1,\ldots,\t$}{
\While{$(\exists \, u \in S_i, \ti(u) \leq N_i(S_\bs)$  \textbf{\em AND}  
$\exists \, v \in V_i-S_i , \ \ti(v)> N_i(S_\bs))$}
{
$S_i=S_i-\{u\}\cup\{v\}$
}}
 \textbf{return} $S_\bs= \cup_{i=1}^{\t} S_i$
\end{algorithm}

The Lemma \ref{lemmaMD} shows that the algorithm ME-ND provides an optimal seed set according to a fixed $\t$-ple $\bs = (s_1,s_2,\ldots,s_\t)$.

\begin{lemma}\label{lemmaMD}
Let $\t$ be the neighborhood diversity of $G$. For any fixed $\t$-ple $\bs = (s_1,s_2,\ldots,s_\t)$, the algorithm ME-ND($\bs,G$) computes a seed set $S_\bs,$ such that $|\I{S_\bs}|$ is maximum among all the seed set $S$ such that each $|S \cap V_i|=s_i$, for $i=1,\ldots,t$. 
\end{lemma}
\proof
Let $S_\bs= \cup_{i=1}^{\t} S_i$ be the seed set returned by the algorithm
ME-ND($\bs,G$).
Let now $S'$ be any optimal seed set satisfying the decomposition  $\bs$, i.e.,  $|\I{S'}|$ is maximum among all the seed set $S$ such that each $|S \cap V_i|=s_i$, for $i=1,\ldots,t$. 
We show that 
{$|\I{S_\bs}|\geq  |\I{S'}|$.} 
	To this aim, 
	 we  iteratively transform each  $S'_i$ into $S_i$ by trading a node $u\in S'_i-S_i$ for a node $v\in S_i-S'_i$
	without decreasing the number of informed nodes.

\begin{itemize}
	\item  If we can choose $v$ such that     $\ti(v) > N_i(S')$  then by  Lemma \ref{lemma1-n-div} we get that 
	$S''=(S'-\{u\}) \cup \{v\}$ has $|\I{S''}| \geq |\I{S'}|$.
	\item  Suppose now that  for any choice of  $v$ it holds  $\ti(v) \leq  N_i(S')$.
	It is possible  to see that  the sets $S_i$  (both as initially chosen at line 2 of the algorithm as
	well as after each update) 
	maximize the number of evangelized nodes in each $V_i$ and 
	$N_i(S_\bs)\geq N_i(\overline{S})$, for any seed set $\overline{S}$ such $|\overline{S}\cap V_i|=s_i$, for $i=1,\ldots,t$. 
	Hence, 
	$$N_i(S_\bs)\geq N_i({S'}), \qquad  \mbox{ for $i=1,\ldots,t.$}$$
	Furthermore, the construction of the sets $S_i$ excludes the possibility that $\ti(u)> N_i(S_\bs)$ and $\ti(v)\leq N_i(S_\bs)$ (cfr. lines 5-6 of the algorithm).
	Therefore, we can assume that $\ti(u)<  N_i(S_\bs)$ and $\ti(v)\leq  N_i(S')\leq N_i(S_\bs)$ for each $u\in S'_i-S_i$  
	and $v\in S_i-S'_i$. In such a case, we have 
	
	$\qquad\quad \I{S'}\cap V_i\subseteq S_i'\cup \left\{ w\in V_i\ |\ \ti(w)\leq N_i(S_\bs)\right\}\subseteq V_i\cap \I{S_\bs}.$

	\noindent
	Hence,  $\I{S'}\subseteq \I{{S_\bs}}$ and we can straight conclude that $|\I{{S_\bs}}| \geq |\I{{S'}}|$.\qed
	\end{itemize}

\begin{theorem}\label{teorema-n-div}
Let $\t$ be the neighborhood diversity of $G$.
It is possible to decide the ($\alpha,\beta$)-\MSS  question 
  in time $O(n \t \ 2^{\t \log (\beta +1) })$.
\end{theorem}

\begin{proof}
For any possible 
 $\bs = (s_1,s_2,\ldots,s_\t)$, denote by $S_\bs$ the output of Algorithm ME-ND($\bs,G$).
We then consider the evangelization process  in $G$ starting at $S_\bs$ and  get the number  $\alpha_\bs= |\I{S_\bs}|$ of  influenced nodes at the end of the process, which, thanks to Lemma \ref{lemmaMD}, is optimal for the partitioning $\bs$.
Finally, we determine $\bs' = {\argmax}_{\bs} \ \alpha_\bs$ and comparing  $\alpha$ with $\alpha_{\bs'}$ we are able to answer the ($\alpha,\beta$)-\MSS  question.

Now we evaluate the running time of the algorithm.
The number of all the possible $\t$-ples $\bs =(s_1,s_2,\ldots,s_\t)$ such that 
$\sum_{i=1}^{\t} s_i =\beta$ is ${\beta + \t -1 \choose \t -1} < 2^{\t \log (\beta +1) }$. 
Moreover, one needs $O(n \t)$ time  to construct $S_\bs$ and  $O(n \t)$ time  to determine $|\I{S_{\bs}}|$.  
Hence, the  time for deciding if a ($\alpha,\beta$)-\MSS for $G$ exists is $O(n \t \ 2^{\t \log (\beta +1) })$. 
\end{proof}

\smallskip

Noticing that  the type partition  $\V$ can be obtained in polynomial time, one has that the 
 ($\alpha,\beta$)-\MSS  problem is  in the class FPT when parameterized by the neighborhood diversity $\t$ and the solution size $\beta$.

\remove{
REMOVE
We stress that our algorithm holds for arbitrary  thresholds. If all the nodes have the same  influence threshold we can  show  that the problem is  FPT with respect to the single parameter $\t$.

\begin{corollary}\label{cor-n-div-const}      
Let $\t$ be the neighborhood diversity of $G$.
Let $\ti(v)=c$, for each $v \in V$.
It is possible to decide the  ($\alpha,\beta$)-\MSS  question   in time $O(n \t \ 2^{\t \log c \t})$.
\end{corollary}

\begin{proof}{}
If $\beta \geq c\t$ then we define $S$ as the subset of $V$ obtained by choosing $c$ nodes from each $V_i \in \V$. By Definition \ref{diversity}, it easy to see that each vertex in $V-S$ has at least $c$ neighbors in $S$.
Hence at the end of the first round of the process  in $G$ starting at $S$ all the nodes in $V$ are influenced. Since $|V| \geq \alpha$, the theorem is proved in this case.
If $\beta < c \t$ then we use Theorem \ref{teorema-n-div} and we have the result.
\end{proof}

{A  more efficient algorithm can be get in case $\ti(v)=c$ for each $v \in V$ and $\beta =O(\t)$.}

\begin{theorem}\label{teorema-n-div-const}    
Let $\t$ be the neighborhood diversity of $G$.
Let $\ti(v)=c$, for each $v \in V$, {and $\beta =O(\t)$.}
It is possible to decide the ($\alpha,\beta$)-\MSS  question  for $G$  in time $2^{O(\t \log \t)}$.
\end{theorem}

\begin{proof} {??Sketch.??}
{Using the hypothesis that $\ti(v)=c$, for each $v\in V$, we have:
\\
{\bf (a)} For each seed set $S$ and $V_i\in\V$ we have that either  $V_i\subseteq \I{S}$ or no node outside $S$ gets informed, that is, $\I{S}\cap V_i=V_i\cap S$. 
\\
{\bf (b)} We can  limit ourselves to search a solution for  $G$  among the seed sets that, inside  each $V_i$ contain the nodes with the highest evangelization thresholds  (that is, if there exists $S$ s.t. $u,v\in V_i$,  $t_E(u)>t_E(v)$, $v\in S$, and $u\not\in S$ then $|\I{S-\{v\}\cup\{u\}}|\geq |\I{S}|$).
 \\
{\bf (c)} We can assume $c\leq \beta$; otherwise  each $v\in V$ has $\ta(v)\geq\ti(v)=c> \beta$ and no node outside the $\beta$ ones in the seed set can get influenced.\\
The above  considerations imply  that 
the algorithm can proceed  as follows.
For each $V_i\in\V$, we fix a subset $V'_i \subseteq V_i$ containing 
$ \min \{|V_i|, \beta+1\}$  vertices of $V_i$ with the highest evangelization thresholds.
We can then look for  a  ($\alpha ,\beta$)-\MSS  for $G$
 by selecting vertices  from the subsets   $V'_i$ only. 
Namely, we determine
  $\overline{\bs} = {\argmax}_{\bs} |\I{S_\bs}|$ where the maximization is over all the vectors $\bs = (s_1,s_2,\ldots,s_\t)$ with
$\sum_{i=1}^{\t} s_i =\beta$ and  $S_\bs= \bigcup_{i=1}^{\t} S_i$ with  each  $S_i$  containing the $s_i$ vertices of $V'_i$ with the highest evangelization thresholds.
If $|\I{S_{\overline{\bs}}}|\geq \alpha$ then we answer {\sc yes} to the  {\sc ($\alpha,\beta$)-\MSS}  question for $G$ and  $S_{\overline{\bs}}$ is the  {\sc ($\alpha,\beta$)-\MSS} for $G$, otherwise we answer {\sc no}.
The number of all the possible $\t$-ple $\bs =(s_1,s_2,\ldots,s_\t)$ such that $\sum_{i=1}^{\t} s_i =\beta$ is upper bounded by  $2^{\t \log (\beta +1) }$. 
%
For each seed $S_{\bs}$, it is possible to show that $O(\beta \t^2)=O(\t^3)$ is sufficient to obtain the number of influenced nodes at the end of the process  in the subgraph $G'$ induced by $\bigcup_{i=1}^t  V'_i$. 
Recalling  (c) 
and   noticing that (a) implies that 
   $V_i\subseteq \Inf_{G}[S_\bs]$ iff $V'_i\subseteq \Inf_{G'}[S_\bs]$ (the other possible case is  $\Inf_{G}[S_\bs]\cap V_i=V_i\cap S=V'_i\cap S=\Inf_{G'}[S_\bs]\cap V_i$), we can easily compute $|\Inf_{G}[S_\bs]|$ from 
	$|\Inf_{G'}[S_\bs]|$. Hence that   the running time of the algorithm is $O( \t^3 \ 2^{\t \log {\beta}})$.  
}
\end{proof}

END REMOVE
}

\medskip

Theorem \ref{teorema-n-div} can be used to also have  FPT linear time algorithms with vertex cover size as parameter  for  ($\alpha,\beta$)-\MSS. 
Indeed, graphs of bounded vertex cover have bounded neighborhood diversity---while 
the opposite is not true since large cliques have  neighborhood diversity  1 \cite{GR}. 
\begin{theorem}\label{v-cover}
Given a vertex cover of $G$ of size $\ell$, it is possible to decide the  ($\alpha,\beta$)-\MSS question in time $O(n (2^\ell+\ell)2^{(2^\ell+\ell) \log \ell })$.
\end{theorem}

\begin{proof}
Let \vc$(G)$ be a vertex cover of $G$ with $|$\vc$(G)|=\ell$.
If $\beta \geq \ell$ then we can use \vc$(G)$ as seed set. 
Indeed, since the nodes in $V - $\vc$(G)$ are independent, 
after one round of the evangelization process in $G$ 
starting at \vc$(G)$ all the nodes in $V$ are evangelist.
Hence, since $|V| \geq \alpha$, we have proved the 
theorem for $\beta \geq \ell$.\\
Let $\beta < \ell$. Since $G$ has vertex cover size $\ell$, it cannot have a type partition with more than $2^\ell+\ell$ sets
\cite{GR}. Hence, we use  Theorem \ref{teorema-n-div} with $t \leq 2^\ell+\ell$  and get the result.
\end{proof}

\subsection{Parameterization of PES with with Treewidth.}

Roughly speaking, the treewidth  measures the 
``tree-likeness'' of a given graph, in particular  any  
tree has  treewidth  1. 
 We generalize  the results given in  \cite{BHLM-11} for the target set selection problem.
We  design an algorithm for the {\sc  Perfect Evangelic Set (PES)} problem 
that runs in $n^{O(w)}$, where $w$ is the treewidth of the input graph.
If all the nodes have the same  influence threshold we obtain  that the problem is FPT.

\begin{definition}\label{def-tw}
A {\em tree decomposition} of a graph G is a pair $(\T , \X )$,
where $\X$ is a family of subsets of $V(G)$, and $\T$ is a
tree over $\X$ , satisfying the following conditions:

1. \  $\cup_{X \in \X} G[X] = G$, and \quad 
2. \ $\forall v \in V (G)$,  $\{X \in \X \ |\ v \in X\}$ is connected in
$T$. \\
A tree decomposition $(\T , \X )$ of a graph $G$ is {\em nice}
if $\T$ is rooted, binary, each node $X \in\X$ has exactly $w$ vertices, 
and is of one of the following three types:

\vspace*{-0.2truecm}
\begin{itemize}
\item  {\em Leaf node.} $X$ is a leaf in $\T$  and consists of $w$ pairwise
non-adjacent vertices of $G$.
\item  {\em Replace node.} $X$ has one child $Y$ in $\T$, s.t.
$X -Y = \{u\}$ and $Y - X = \{v\}$ for  
$u \neq v$.
\item  {\em Join node.} $X$ has two children $Y$ and $Z$ in
$\T$ with $X = Y = Z$.
\end{itemize}
\vspace*{-0.2truecm}

The {\em width} of T is $\max_{X \in \X}|X| -1$. The {\em treewidth}
of $G$ is the minimum width over all tree (nice) decompositions
of $G$.
\end{definition}

\noindent
The algorithm follows a dynamic programming approach computing a table, for each node $X$ of a nice tree decomposition of $G$, that depends on the pair of thresholds of the vertices in $X$. 
 Each entry in the table stores the smallest seed set for the subgraph $G[X]$ of $G$ induced by the vertices of the subtree rooted at $X$.
The desired seed set for $G$ is the one corresponding to the root node of the tree decomposition of $G$.
The proof follows the lines of the one   in \cite{BHLM-11} for the target set selection problem
(e.g. in the special case $\ta=\ti$), except for the role played by vertices that need to be influenced but not
 evangelized and by the influence thresholds in computing the entries of the table for each node $X$. 
We can prove the following result whose  proof is omitted since, as said before, 
it is essentially patterned after the arguments of Section 3 of \cite{BHLM-11}.

\begin{theorem}\label{teo-tw}
In graphs of treewidth $w$ the  PES problem can be solved in $n^{O(w)}$ time.
\end{theorem}

\remove{

Roughly speaking, the treewidth  measures the 
``tree-likeness'' of a given graph, in particular  any  
tree has  treewidth  1. 
 We generalize  the results given in  \cite{BHLM-11} for the target set selection problem.
We  design an algorithm for the {\sc  Perfect Evangelic Set (PES)} problem 
that runs in $n^{O(w)}$, where $w$ is the treewidth of the input graph.
If all the nodes have the same  influence threshold we obtain  that the problem is FPT.

\begin{definition}\label{def-tw}
A {\em tree decomposition} of a graph G is a pair $(\T , \X )$,
where $\X$ is a family of subsets of $V(G)$, and $\T$ is a
tree over $\X$ , satisfying the following conditions:

1.  $\cup_{X \in \X} G[X] = G$, and 

2.  for all $ v \in V (G)$,  the set $\{X \in \X \ |\ v \in X\}$ is connected in
$T$. 

\noindent
A tree decomposition $(\T , \X )$ of a graph $G$ is {\em nice}
if $\T$ is rooted, binary, each node $X \in\X$ has exactly $w$ vertices, 
and is of one of the following three types:
\begin{itemize}
\item  {\em Leaf node.} $X$ is a leaf in $\T$  and consists of $w$ pairwise
non-adjacent vertices of $G$.
\item  {\em Replace node.} $X$ has one child $Y$ in $\T$, s.t.
$X -Y = \{u\}$ and $Y - X = \{v\}$ for  
$u \neq v$.
\item  {\em Join node.} $X$ has two children $Y$ and $Z$ in
$\T$ with $X = Y = Z$.
\end{itemize}
The {\em width} of T is $\max_{X \in \X}|X| -1$. The {\em treewidth}
of $G$ is the minimum width over all tree (nice) decompositions
of $G$.
\end{definition}
The algorithm follows a dynamic programming approach computing a table, for each node $X$ of a nice tree decomposition of $G$, that depends on the pair of thresholds of the vertices in $X$. 
 Each entry in the table stores the smallest seed set for the subgraph $G[X]$ of $G$ induced by the vertices of the subtree rooted at $X$.
The desired seed set for $G$ is the one corresponding to the root node of the tree decomposition of $G$.
The proof follows the lines of the one   in \cite{BHLM-11} for the target set selection problem
(e.g. in the special case $\ta=\ti$), except for the role played by vertices that need to be influenced but not
 evangelized and by the influence thresholds in computing the entries of the table for each node $X$. 
We can prove the following result.

\begin{theorem}\label{teo-tw}
In graphs of treewidth $w$ the  PES problem can be solved in $n^{O(w)}$ time.
\end{theorem}

\def\OGXd{OPT_{G_X^d}}
Given a nice tree decomposition $(\T , \X )$ of $G$, we will assume
that $\T$ is rooted at some arbitrary $R \in \X$.  For a node
$X \in \X$ , let $\T_X$ denote the subtree of $\T$ rooted at $X$,
and let $\X_X$ denote the collection of nodes in this tree,
including $X$ itself. The subgraph $G_X$ associated with
$X$ in $\T_X$ is defined by $\GX = \cup_{Y \in \X_X} G[Y]$.
The vertices of $X$ are called the {\em boundary} of $\GX$.\\
In order to describe our algorithm we need some further definition. 
Let $[n] = \{0, 1, \ldots , n\}$.
\begin{definition}
An {\em evangelization threshold vector}, $\taa \in [n]^w$,
(resp. {\em influence threshold vector}, $\tii \in [n]^w$)
 is a vector
with a coordinate for each boundary vertex in $X$.
Let  $\taa(v)$ (resp. $\tii(v)$) denote the coordinate in $\taa$ (resp. $\tii(v)$) corresponding
to the boundary vertex $v \in X$, and $\ta$ (resp. $\ti$)
denote the original
evangelization threshold (resp. influence threshold) function of $G$, the subgraph $\GX(\taa,\tii)$ is defined
as the graph $\GX$ with  thresholds:\\
$\bullet$ \  $\taa(v)$ and $\tii(v)$ for any boundary vertex $v \in X$, and\\
$\bullet$ \  $\ta(u)$ and $\ti(v)$ for all other vertices $u \not \in X$.
\end{definition}
\begin{definition}
An {\em evangelization order} is a function
$\IA : X \rightarrow \{0,1, \ldots,w-1, w\}$, that for any $v \in X$, is defined as follows: \\
- if $\IA(v) = i$, for $0 \leq i \leq w-1$, then $v$ is active and $i$ denotes the relative iteration in the boundary $X$ at which $v$ is activated;\\
- if $\IA(v) = w$ then  $v$ is influenced but not an evangelist.
\end{definition}
Given a seed set $S \subseteq V(\GX)$ and an evangelization order 
$\IA$, the $\IA$-{\em constrained evangelization process of $S$ in $\GX(\taa,\tii)$} is
defined similarly to the normal evangelization process of
$S$ in $\GX(\taa,\tii)$, except that a boundary vertex $v$ becomes an evangelist
at iteration i only if all boundary vertices $u$ with
$\IA(u) < \IA(v)$ are evangelists at iteration $i-1$,
and only if all other boundary vertices $w$ with $\IA(w) = \IA(v)$ 
will also become evangelized at  iteration $i$.
 This includes all boundary vertices $u \in S$.
 Note that all vertices that are evangelized by $S$ and so all the vertices that are influenced by $S$ in a normal evangelization process get evangelized and influenced
 in an $\IA$-constrained process for some evangelization order $\IA$.

A seed set $S \subseteq V(\GX)$, with $|S| \leq \beta$, that influences {\em all vertices} in $\GX(\taa,\tii)$ in a $\IA$-constrained evangelization process is said a 
{\em perfect seed set conforming to $\IA$}.

Our algorithm computes for each graph $\GX$, corresponding to node $X$ of $\T$, $\OGX$ that depends on the influence and evangelization threshold vectors $\tii,\taa$ of the boundary vertices of $\GX$ and on the evangelization order $\IA$ on the boundary vertices.
In particular, the entry $\OGX[\taa,\tii,\IA]$ stores the smallest seed set for $\GX(\taa,\tii)$ conforming to $\IA$.\\
We stress that  some evangelization order $\IA$ could  {\em  not  be valid}; this occurs when for some boundary vertex $u$ with 
$Ne_X(\IA,u)= | \{z |\  z\  \mbox{is a neighbor of $u$ in $\GX$ and} \ \IA(z) \leq w-1 \}|$
 it holds that:\\
- $0 \leq \IA(u) \leq w-1$ and $ Ne_X(\IA,u)< \taa(u)$, i.e., $u$ should be an evangelist but the number of its evangelist neighbors in $\GX$ is less than $\taa(u)$;\\
- $\IA(u) = w$ and either $Ne_X(\IA,u)\geq \taa(u)$ or $Ne_X(\IA,u)< \tii(u)$, i.e., $u$ should be influenced but not an evangelist and the number of its evangelization  neighbors in $\GX$ is either at least $\taa(u)$ or less than $\tii(u)$.\\
We assume that $\OGX[\taa,\tii,\IA]$ is not taken in consideration when $\IA$ is not valid.

\begin{lemma}\label{tabella}
Let $t_{max}$ be the maximum evangelization and/or influence threshold of any vertex in the network.
The number of different entries in $\OGX$ is upper bounded by $t_{max}^{O(w)}\leq n^{O(w)}$.
\end{lemma}
Recall that $G_{R} = G$ where $R$ is the root of $\T$. 
Therefore, if we compute the $OPT_{G_R}$
table for the root $R$ we can
determine the optimal perfect seed set for $G$. 
The algorithm will compute the $OPT_{G_X}$ tables in bottom-up fashion, where the computation at the leaves will be done by brute-force.
According to Lemma \ref{tabella} , and since $\T$
has $O(wn)$ nodes, to obtain time bound in Theorem \ref{teo-tw}
all that is required is to compute  $\OGX$ for each  $X \in \X$ in $n^{O(w)}$ time.
Since the graphs at the leaves only have $w$ vertices, this can be done in $n^{O(w)}$ time for a leaf node $X$.  In the following we give details on how to compute $\OGX$ table in case $X$ is an internal node of $\T$ from the table(s) correspond to its child(ren) in $\T$ . 

\noindent
\underline{{\em Leaf Nodes.}}
If $X$ is a leaf in $\T$, then $\GX$ is a graph with $w$ isolated vertices,
and so any perfect seed set for $\GX$ must include
all vertices with  influence  threshold greater than 0.
Furthermore, all vertices will get evangelized regardless of the 
evangelization order we impose on the boundary. 
Hence, for each leaf node $X$ we have
$$\OGX [\taa,\tii,\IA] = X - \{v \ |\  \tii(v) = 0\}.$$

\smallskip

\noindent
\underline{{\em Replace Nodes.}}
Suppose $X$ is a replace node with
child $Y$ in $\T$. That is, $\GX$ is obtained by adding a
new boundary vertex $u$ to $\GY$, and removing another
boundary vertex $v$ from the boundary (but not from
$\GX$). Recall that by the second condition of Definition \ref{def-tw}, $u$ can
only be adjacent to other boundary vertices of $\GX$.
To compute $\OGX$, we need to distinguish two cases according to the value of the evangelization order $\IA$ related to vertex $u$.\\
$\bullet$ \ Let $0 \leq \IA(u) \leq w$ and $ Ne_X(\IA,u) \geq \taa(u)$.
In this case the order $\IA$ requires that  $u$ is an evangelist.
Let $d$ denote the number of these neighbors of $u$ in $\GX$,
and assume that they are ordered. Also, let $\GX^i$
for $i = 0, \ldots, d$ denote the subgraph of $\GX$ obtained by
adding the edges between $u$ and  all of its neighbors
in $X$, up-to and including the $i$-th neighbor. To compute
$\OGX$, we will actually compute $\OGXi$
in increasing values of $i$, letting $\OGX= \OGXd$.

When $i = 0$, $u$ is isolated, and thus it must be included
in any perfect seed set when it has evangelization threshold 
greater than 0. For any  threshold vector $\tii$ (resp. $\taa$), let 
$\tii^{uv}$ (resp. $\taa^{uv}$) denote the influence
(resp. evangelization) threshold vector  obtained by setting: 
$\tii^{uv}(w) = \tii(w)$  (resp. $\taa^{uv}(w) = \taa(w)$) for all $w \neq u$, and 
$\tii^{uv}(v) = \tii(u)$ (resp. $\taa^{uv}(v) = \taa(u)$). 
For an order $\IA$ for $X$, let $\IA^{uv}$ denote the set of all orderings 
$\IA'$ for $Y$ with $\IA'(w) = \IA(w)$ for all boundary vertices $w \neq u, v$. 
Observe that  $\IA'(v) \neq \IA(u)$ is allowed.
According to the above, when $X$ is a replace node, we
get for $i = 0$:
$$\OGXo[\taa,\tii,\IA] = \min_{\IA' \in \IA^{uv}}
\begin{cases}
{\OGY[\taa^{uv},\tii^{uv},\IA']}& {\mbox{if $\taa(u)=0$}}\\
{\OGY[\taa^{uv},\tii^{uv},\IA']\cup \{u\}} & {\mbox{if $\taa(u)\neq 0$}}   
\end{cases}$$

If $i>0$ then $\GX^i$ is obtained from $\GX^{i-1}$ by connecting $u$ to some boundary vertex $z \in X$.
For any  threshold vector $\tii$ (resp. $\taa$), let $\tii^{u-}$  (resp.  $\taa^{u-}$) denote the  threshold vector obtained by setting $\tii^{u-}(u) = 
\max \{\tii(u) - 1, 0\}$ (resp.  $\taa^{u-}(u) = \max \{\taa(u) - 1, 0\}$),  and all remaining thresholds the same. Define $\tii^{z-}$ (resp. $\taa^{z-}$) similarly.
Since the edge $(u, z)$ can only influence $v$ if $\IA(z) < \IA(v)$, and vice-versa,
we have:
$$\OGXi[\taa,\tii,\IA] = \begin{cases}
{\OGXii[\taa,\tii,\IA]}& {\mbox{if $\IA(z) \leq w-1$ and $\IA(z)=\IA(u)$}}\\
{\OGXii[\taa^{u-},\tii,\IA]}& {\mbox{if $\IA(z) \leq w-1$ and $\IA(z)<\IA(u)$}}\\
{\OGXii[\taa^{z-},\tii,\IA]} & {\mbox{if $\IA(z) \leq w-1$ and $\IA(z)>\IA(u)$}}\\
{\OGXii[\taa,\tii^{z-},\IA]} & {\mbox{if $\IA(z) = w$.}}     
\end{cases}$$
$\bullet$ \ Let $\IA(u)=w$ and $\tii(u) \leq Ne_X(\IA,u)< \taa(u)$.
In this case the order $\IA$ requires that  $u$ is influenced but not an evangelist.
For an order $\IA$ for $X$, let $\IA^{uv}$ defined as in the above case.
We have
$$\OGX[\taa,\tii,\IA] =\min_{\IA' \in \IA^{uv}} \OGY[\taa^{uv},\tii^{uv},\IA'].$$

\smallskip

\noindent
\underline{{\em Join Nodes.}}
Let $X$ be a join node with children $Y$
and $Z$ in $\T$. Recall that $\GY$ and $\GZ$ are two subgraphs
whose intersection is exactly their boundary $Y = Z$, and
$\GX$ is obtained by taking the union of these two subgraphs.
Observe that this means that there are no edges between $V(\GY )- Y$
 and $V(\GZ )- Z$ in $\GX$. 
For a boundary vertex $v \in X$, let $N_X(v)$ denote
the set of boundary vertices that are connected to $v$ in
$\GX$. For $v \in X$, and an evangelization order $\IA$, we define
$\IA^{<v}$ to be the set of all boundary vertices $u$ such that
$\IA(u) < \IA(v)$.
For the influence threshold vectors $\tii^Y$ , $\tii^Z$, the evangelization threshold vectors $\taa^Y$ , $\taa^Z$ and an
evangelization order $\IA$, we define the threshold vectors 
$\taa(v)=\taa^Y \oplus_\IA \taa^Z (v)= \taa^Y(v)+\taa^Z(v) -|N_X(v) \cap \IA^{<v}|$ and
$\tii(v)=\tii^Y \oplus_\IA \tii^Z (v)= \tii^Y(v)+\tii^Z(v) -|N_X(v) \cap \IA^{<v}|$
for every $v \in X$. 
To compute $\OGX [\taa,\tii, \IA]$ we first compute
$$[(\overline{\taa^Y},\overline{\taa^Z}), (\overline{\tii^Y}, \overline{\tii^Z})] = \argmin_{\stackrel{\taa^Y \oplus_\IA \taa^Z=\taa}{\tii^Y \oplus_\IA \tii^Z=\tii}} |\OGY[\taa^Y,\tii^Y,\IA] \cup \OGZ[\taa^Z,\tii^Z,\IA]|$$
then we obtain
$ \OGX [\taa,\tii,\IA]= \OGY[\overline{\taa^Y},\overline{\tii^Y},\IA] \cup \OGZ[\overline{\taa^Z},\overline{\tii^Z},\IA].$
\\
To see the correctness of the above equation we consider that any
perfect seed set $S$ for $\GX(\taa,\tii)$ which conforms to 
$\IA$ can be decomposed into two subsets $S_Y = S \cap V(\GY)$
and $S_Z = S\cap V(\GZ)$ which influence in an $\IA$-constrained
evangelization process all vertices in $\GY(\taa^Y,\tii^Y)$ 
and $\GZ(\taa^Z,\tii^Z)$,
for some pair of evangelization threshold vectors $\taa^Y, \taa^Z$ and some pair of influence threshold vectors $\tii^Y, \tii^Z$ for which
$\taa(v)=\taa^Y \oplus_\IA \taa^Z (v)$ and 
$\tii(v)=\tii^Y \oplus_\IA \tii^Z (v)$. 
The converse is also true; any pair
of perfect seed sets for $\GY(\taa^Y,\tii^Y)$ 
and $\GZ(\taa^Z,\tii^Z)$ conforming
to $\IA$, can be united into a perfect seed set for
$\GX(\taa^Y \oplus_\IA \taa^Z,\tii^Y \oplus_\IA \tii^Z)$, also conforming to $\IA$.

}

\section{ Exact Polynomial Time Algorithms for MES}\label{sec-trees}
{In this section we show that the MES problem is exactly solvable in polynomial time on complete graphs and trees.}

\subsection{Complete Graphs}

Since the neighborhood diversity of a complete graph is $1$ we already know that the \MSS problem is  solvable in polynomial time
 on complete graphs.  
However, by observing that  when $t=1$, then  $\bs =(s_1)$ is a singleton and there a single $1$-tuple available 
(i.e., $s_1=\beta$), we can design an  algorithm to solve the \MSS problem 
that is is much simpler than the one described in Section \ref{parametr}. 
We show below the MES-K algorithm that represents  a specialized, and more efficient, 
version  of the  ME-ND algorithm to complete graphs.
By Lemma \ref{lemmaMD}, that gives the correcteness of the algorithm, we can prove the following Theorem.
\begin{theorem}\label{teoKfrac} 
 In a complete graph with $n$ nodes, the \MSS problem can be solved 
in 
$O(n)$ time.
\end{theorem}
  
%
%

\begin{algorithm}
\SetCommentSty{footnotesize}
\SetKwInput{KwData}{Input}
\SetKwInput{KwResult}{Output}
\DontPrintSemicolon
\caption{ \ \   \textbf{Algorithm} \MSS-K($K,\beta$) \label{algk}}
\KwData { A clique $K=(V,E)$, threshold functions  $\ti$ and $\ta$,   budget $\beta\leq|V|$. }
\KwResult{ $S$ a seed set for $K$ such that $|S|\leq \beta.$ } 
\setcounter{AlgoLine}{0}
	 Let $X=\{v_{1},v_{2},\ldots,v_{\beta}\}$ be a set of $\beta$ nodes of $V$ with the largest evangelization thresholds (i.e., for any $u\in X$ and $v \in V-X$ it holds $t_E(u)\geq t_I(v)$) and  $\eta^*=|\A{X}|$\\
	Set $S=X$\\
	
	\While{$(\exists \, u \in S, \ \ti(u) \leq \eta^*$  \textbf{\em AND}  
$\exists \, v \in V-S , \ \ti(v)> \eta^*)$}
{
$S = S - \{{u}\}\cup\{v\}$
}
return $S$
\end{algorithm}

\subsection{Trees}

{Thanks to Theorem \ref{teo-tw}, we know that the PES problem is solvable in polynomial time on graphs having constant treewidth. In the special case of trees, we are able to solve in polynomial time also the MES problem. In the following we  give a dynamic programming algorithm that proves Theorem \ref{theorem-tree}.}

\begin{theorem}\label{theorem-tree}
The  \MSS problem with bound $\beta$ can be 
solved in  time $O(\min\{n\Delta^2\beta^3,n^2\beta^3\})$ on  any tree with $n$ nodes and maximum degree $\Delta$.
\end{theorem}

\noindent
 The rest of this section is devoted to the description and analysis of the
  algorithm proving   Theorem \ref{theorem-tree}.
	Let $T = (V,E)$ be a tree rooted  at any  node $r$ and 
denote by $T(v)$ the  subtree rooted at $v$, for   $v\in V$.
The algorithm   makes a postorder traversal of  the input  tree $T$.
 For each node $v$, the algorithm solves all possible   instances of the 
{\MSS} problem  on the subtree $T(v)$, with bound   $b\in\{0,1,\ldots, \beta\}$.  
Moreover, in order to compute these values 
one has to consider---for the root node $v$ of $T(v)$---not only the original thresholds $\ti(v)$ and $\ta(v)$ of $v$, but also 
the decremented  values $\ti(v)-1$ and $\ta(v)-1$ which we call the {\em residual thresholds}.
%
For each node $v\in V$ and  integer $b\geq 0$ we define the following quantities:
%
\begin{eqnarray}
\label{eq-case1} \n{v}{b}&&\mbox{is the maximum number of nodes that can be influenced in $T(v)$, }
\\ \nonumber 
&&\mbox{assuming that at most $b$ of the nodes  in $T(v)$ belong to the seed set,}
\\ \nonumber
&& \mbox{  if $v$ is still agnostic  at the end of the evangelization process;} 
\end{eqnarray}
\begin{eqnarray}
\label{eq-case2}  \nn{v}{b}&&\mbox{is the maximum number of nodes that can be influenced in $T(v)$}
\\ \nonumber 
&& \mbox{assuming that at most $b$ of the nodes  in $T(v)$ belong to the seed set,}
\\ \nonumber
&& \mbox{ if, at the end of the  process,  $v$ is influenced but it is not an evangelist;} 
\end{eqnarray}
\begin{eqnarray}
\label{eq-case3}  \nnn{v}{b}&&\mbox{is the maximum number of nodes that can be influenced in $T(v)$}
\\ \nonumber 
&& \mbox{assuming that at most $b$ of the nodes  in $T(v)$ belong to the seed set,}
\\ \nonumber
&& \mbox{  if $v$ is an evangelist at the end of the evangelization process.} 
\end{eqnarray}
Similarly the quantities $\nwI{v}{b}$, $\nnwI{v}{b}$ and $\nnnwI{v}{b}$ represent the same quantities as above but considering the decreased thresholds for $v$ (which may reflect the fact that the parent node of $v$  becomes an evangelist before $v$ itself).
\\
We define the above  quantities be $-\infty$ if any of the constraints is not satisfiable.
For instance, if $v$ is a single node, $b=0$ and $\ta(v)>0$ we set\footnote{Indeed $v$ should be
 an evangelist,  however 
 the budget is $0$ while  the threshold is $>0$.} $\nnn{v}{0}= -\infty$.

\begin{remark}
We  mention that  all the above quantities  are monotonically non-decreasing in $b$  and that $\n{v}{b} \leq \nwI{v}{b}$, $\nn{v}{b} \leq \nnwI{v}{b}$ and $\nnn{v}{b} \leq \nnnwI{v}{b}$.
\end{remark}

The maximum number of nodes in $T$ that can be influenced with any seed set of size  $\beta$
can be then obtained by computing
\begin{equation}\label{eq-mas}
\max\{\n{r}{\beta},\, \nn{r}{\beta},\, \nnn{r}{\beta} \}.
\end{equation}
In order to obtain the value in (\ref{eq-mas}), we compute the quantities\footnote{For the root node $r$, the quantities $\nwI{r}{b}$, $\nnwI{r}{b}$ and $\nnnwI{r}{b}$ are not required.} $\n{v}{b}$, $\nn{v}{b}$, $\nnn{v}{b}$, $\nwI{v}{b}$, $\nnwI{v}{b}$ and $\nnnwI{v}{b}$ for each $v \in V$ and  for each $b=0,1,\ldots,\beta$.

We  proceed postorder fashion on the tree, so that the computation of the various values  for a node $v$ is done after all  the values for $v$'s children are known.

\smallskip
For each leaf node $\ell$ we have the  values below. 
Recall that they refer to the tree
$T(\ell)$ consisting of the single node $\ell$.\\
The node  $\ell$ will be not even influenced only if the budget is not sufficient  to have  $\ell$ in the seed set (e.g. $b=0$) while  the influence threshold is  
$\ti(\ell)>0$. Hence,
\begin{equation}\label{eq-casel}
\n{\ell}{b}=  \begin{cases} 0 & \mbox{ if }  (b=0 \mbox{ AND } \ti(\ell)>0)\\
-\infty &  \mbox{otherwise,} \end{cases} 
\end{equation}
The node $\ell$ gets influenced but  does not become an evangelist   in case the budget is not 
sufficient to  have  $\ell$ in the seed set (e.g. $b=0$) and the evangelization threshold is  
  $\ta(\ell)>0$, but   the influence threshold is  $\ti(\ell)=0$. Hence, 
\begin{equation}\label{eq-casel-1}
\nn{\ell}{b}= \begin{cases}1 & \mbox{ if }  (b=0 \mbox{ AND } \ti(\ell)=0 \mbox{ AND }\ta(\ell)>0)\\
-\infty &  \mbox{otherwise.} \end{cases}   
\end{equation}
The node $\ell$ becomes  evangelist in $T(\ell)$ when either the budget is sufficiently large  to have $\ell$ in  the seed set ($b\geq 1$) or  its evangelization  threshold 
is  $\ta(\ell)=0$. Hence,
\begin{equation}\label{eq-casel-2}
\nnn{\ell}{b}= \begin{cases} 1 & \mbox{ if }  (b\geq 1 \mbox{ OR } \ta(\ell)=0)\\
-\infty &  \mbox{otherwise.} \end{cases}   
\end{equation}
The values for $\nwI{\ell}{b}$, $\nnwI{\ell}{b}$ and $\nnnwI{\ell}{b}$ are computed similarly by using 
on $\ell$ the residual thresholds ($\ti (\ell)-1$ and $\ta(\ell)-1$) instead of
$\ti(\ell)$ and $\ta(\ell)$.

\medskip
We show now that for any internal node $v$    and for any  integer $b\in\{0,\ldots,\beta\}$,  each  of the values $\n{v}{b}$, $\nn{v}{b}$, $\nnn{v}{b}$, $\nwI{v}{b}$, $\nnwI{v}{b}$, and $\nnnwI{v}{b}$  can be computed   in time $O(d^2 b^2)$, where $d$ is the number of children of $v$ in $T$.

We recall that when computing one of the values
$\n{v}{b}$, $\nn{v}{b}$, $\nnn{v}{b}$, $\nwI{v}{b}$, $\nnwI{v}{b}$ or $\nnnwI{v}{b}$, we already have computed all the values for each child $v_i$ of $v$.
We  distinguish two cases: The computation of  the values $\n{v}{b}$ and $\nn{v}{b}$ and the computation of the values $\nnn{v}{b}$.

\noindent \underline{\bf 1. Computation of  $\n{v}{b}$ and $\nn{v}{b}$}.
In this case we  know that $v$ will not become evangelist. Hence, we do not use  
the budget for the node $v$ itself and the computation of 
$\n{v}{b}$ and $\nn{v}{b}$ must  consider all the possible ways in which the whole budget $b$ can be partitioned among $v$'s children.
\begin{fact}\label{lemma1}
It is possible to compute $\n{v}{b}$, $\nn{v}{b}$, $\nwI{v}{b}$ and $\nnwI{v}{b}$, in time $O(d^2 b^2),$ where $d$ is the number of children of $v$.
\end{fact}

\begin{proof}
We focus our attention on $\n{v}{b}$ and $\nn{v}{b}$, the remaining values can be computed in the same way but for  using the residual threshold $\ta(v)-1$ and $\ti(v)-1$ on $v$ instead of original ones.
\\
Fix an ordering $v_1,v_2,\ldots, v_d$ of the children of node $v$.
For $i=1,\ldots, d$, $j=0,\ldots, b$ and $k=0, \ldots, d$, let $\Amax_v[i,j,k]$ be the maximum number of nodes that can be influenced  in  the forest consisting of the
 (sub)trees 
 $T(v_1),T(v_2),\ldots,T(v_i)$,  assuming  that these trees contain  at most $j$ seeds altogether  and that $k$  among
their roots $v_1,v_2,\ldots,v_i$ will  
become evangelist---in the respective tree.
By (\ref{eq-case1}) and (\ref{eq-case2})  we have 
\begin{equation}\label{eq-amax}
\n{v}{b}=\max_{k\in\{0,1,\ldots,\ti(v)-1\}}\Amax_v[d,b,k]
\end{equation}

\begin{equation}\label{eq-amax2}
\nn{v}{b}=\max_{k\in\{\ti(v),\ti(v)+1,\ldots,\ta(v)-1\}}\Amax_v[d,b,k]+1.
\end{equation}
We now show how to compute $\Amax_v[d,b,k]$ for $k\in\{0,1,\ldots,\ta(v)-1\}$ by recursively computing the values $\Amax_v[i,j,k]$, for each $i=1,2,\ldots,d$,  $j=0,1,\ldots,b$ and  $k=0, \ldots, \ta(v)-1$.

For $i=1$, we  assign all of the budget to $T(v_1)$ and 
$$\Amax_v[1,j,k]=
\begin{cases}
 \max \{\n{v_1}{j},\nn{v_1}{j}\} & \mbox{if  $k=0$} \\
\nnn{v_1}{j} & \mbox{if  $k=1$} \\
-\infty & \mbox{if  $k>1.$} 
\end{cases}
$$
For $i>1$, we consider 
each $0\leq a \leq j$:  Budget $a$ is assigned to the first $i-1$ trees, while the remaining budget $j-a$ is assigned to $T(v_i)$. Hence,  
$$   \Amax_v[i,j,k]= \max
\begin{cases}
\max_{0\leq a \leq j} \left \{ \Amax_v[i{-}1,a,k] + \max \{\n{v_i}{j-a},\nn{v_i}{j-a}\}\right \}\\
\max_{0\leq a \leq j} \left \{ \Amax_v[i{-}1,a,k{-}1] + \nnn{v_i}{j-a}\right\} 
\end{cases}
$$ 
The computation of  $\Amax_v[\cdot,\cdot,\cdot]$ involves  $O(d^2  b)$ values, each   recursively computed
in time $O(b)$.  Hence 
we are able to compute it---and 
by (\ref{eq-amax}) and (\ref{eq-amax2}) ,  also $\n{v}{b}$ and $\nn{v}{b}$---in time $O(d^2 b^2)$.
\end{proof}

\noindent \underline{\bf 2. Computation of  $\nnn{v}{b}$.}
We focus our attention on $\nnn{v}{b}$, the same reasoning applies  to $\nnnwI{v}{b}$ by using the residual threshold on $v$ instead of the original one.
In this case we   know that $v$ will be  an evangelist and we have two cases to consider depending whether $v$  belongs to the seed set or not. In the following we will analyze the two cases separately. The desired value will be 
\begin{equation} \label{M2}
\nnn{v}{b} =\max
\{
M_1,
M_2 
\}, 
\end{equation}
where   $M_1$ denotes  the value one obtains assuming     $v\in S$ and by
$M_2$ denotes  the value one obtains assuming     $v\notin S$.
\begin{itemize}
\item $v\in S$. 
In this case we assume that $\ta(v)>0$ (otherwise $v$ would  become an evangelist anyhow and it makes no sense to spend part of the budget to evangelize it).
We consider $b\geq 1$ (otherwise $M_1=-\infty$). Since  $v\in S$  the computation of 
$M_1$ must  consider all the possible ways in which the remaining budget $b-1$ can be partitioned among $v$'s children.
\begin{fact}\label{lemma3}
$M_1$ is computable  in time $O(d b^2),$ where $d$ is the number of children of $v$.
\end{fact}
\begin{proof}
Fix an ordering $v_1,v_2,\ldots, v_d$ of the children of node $v$.
For $i=1,\ldots, d$ and $j=0,\ldots, b$  let $\Cmax_v[i,j]$ be the maximum number of nodes that can be influenced  in the  first $i$ subtrees $T(v_1),T(v_2),\ldots,T(v_i)$ assuming that the seed set contains  $v$ and at most $j$ among  the nodes  in such  subtrees.
By (\ref{eq-case3}) we have 
\begin{equation}\label{eq-cmax}
M_1=\Cmax_v[d,b-1]+1.
\end{equation}
We now show how to compute $\Cmax_v[d,b-1]$ by recursively computing the values $\Cmax_v[i,j]$, for each $i=1,2,\ldots,d$ and  $j=0,1,\ldots,b-1$.

For $i=1$, we  assign all of the budget to $T(v_1)$ and 
$$\Cmax_v[1,j]=\max \left\{ \nwI{v_1}{j}, \nnwI{v_1}{j},\nnnwI{v_1}{j}\right\}. 
$$
For $i>1$, we consider  each $a\in \{0, \ldots, j\}$ and assign  budget $a$ to the first $i-1$ subtrees, while the remaining  budget $j-a$ is assigned to $T(v_i)$. Hence,  

$$\Cmax_v[i,j]= \max_{0\leq a \leq j} \Big\{ \Cmax_v[i{-}1,a] + \max 
    \left\{  
    \nwI{v_i}{j-a}, \nnwI{v_i}{j-a},\nnnwI{v_i}{j-a}
    \right\}\Big\}.$$
				
\noindent
The computation of  $\Cmax_v$ uses  $O(d  b)$ values and each one is computed recursively in time $O(b)$.  Hence, 
we are able to compute it  and, 
by (\ref{eq-cmax}),   $M_1$, in time $O(d b^2)$.
\end{proof}

\item  $v\notin S$.
In this case we  know that $v$ will be made an evangelist by the evangelic action of  
(some of) its children. Hence the computation of 
$M_2$ must  consider all the possible ways in which the (whole) budget $b$ can be partitioned among $v$'s children in such a way that at least  $\ta(v)$ of $v$'s children become evangelists.
\begin{fact}\label{lemma4}
 $M_2$ can be computed in time $O(d^2 b^2),$ where $d$ is the number of children of $v$.
\end{fact}
\begin{proof}{}
Fix an ordering $v_1,v_2,\ldots, v_d$ of the children of the  node $v$.
For $i=1,\ldots, d$, $j=0,\ldots, b$, and $k=0, \ldots, d$, let $\Dmax_v[i,j,k]$ be the maximum number of nodes that can be influenced, in $T(v_1),T(v_2),\ldots,T(v_i)$ assuming  that: $v$ will be an evangelist,
at most $j$ among  the nodes in   $T(v_1),\ldots,T(v_i)$ belong to the seed set,  and $k$  among $v_1,v_2,\ldots,v_i$ will be evangelists
(in the respective subtrees).
By (\ref{eq-case3}) we have 
\begin{equation}\label{eq-dmax}
M_2=\max_{k\geq \ta(v)}\Dmax_v[d,b,k]+1.
\end{equation}
We now show how to compute $\Dmax_v[d,b,k]$ for $k\in\{\ta(v),\ta(v)+1,\ldots,d\}$ by recursively computing the values $\Dmax_v[i,j,k]$, for each $i=1,2,\ldots,d$,  $j=0,1,\ldots,b$ and  $k=0, \ldots, d$.

For $i=1$, we  assign all of the budget to $T(v_1)$ and 
$$\Dmax_v[1,j,k]=
\begin{cases}
 \max \left \{\nwI{v_1}{j},\nnwI{v_1}{j}, \nnnwI{v_1}{j} \right \} & \mbox{if  $k=0$} \\
\nnn{v_1}{j} & \mbox{if  $k=1$} \\
-\infty & \mbox{if  $k>1.$} 
\end{cases}
$$
Consider now $i>1$. For each $a\in \{0,\ldots,j\}$ we  assign budget $a$ to 
the first $i-1$ subtrees, while the remaining  budget $j-a$ is assigned to $T(v_i)$. Hence,  
$$   \Dmax_v[i,j,k]= \max
\begin{cases}
\max_{0\leq a \leq j} \Big \{ \Dmax_v[i{-}1,a,k] + \\
\qquad\qquad\qquad \max \{\nwI{v_i}{j-a},\nnwI{v_i}{j-a},\nnnwI{v_i}{j-a}\} \Big \}\\
\max_{0\leq a \leq j} \left \{ \Dmax_v[i{-}1,a,k{-}1] + \nnn{v_i}{j-a}\right\} 
\end{cases}
$$ 
The computation of  $\Dmax_v$ comprises  $O(d^2  b)$ values and each one is computed recursively in time $O(b)$.  Hence 
we are able to compute it, and 
by (\ref{eq-dmax}),  also $M_2$, in time $O(d^2 b^2)$.
\end{proof}

As a consequence of  Facts \ref{lemma3} and \ref{lemma4} and equation \ref{M2}, we are able to compute $\nnn{v}{b}$ and $\nnnwI{v}{b}$, in time $O(d^2 b^2)$.
\end{itemize}

%

\noindent
The above Facts \ref{lemma1}-\ref{lemma4}, imply  that  the   value
$\max\{\n{r}{\beta},\, \nn{r}{\beta},\,  \nnn{r}{\beta} \}$   in (\ref{eq-mas}) can be computed in time
$$\sum_{v \in V} O(d(v)^2\beta^2) {\times} O(\beta)=O(\beta^3)\times\sum_{v \in V} O(d(v)^2)=O(\min\{n\Delta^2\beta^3,n^2\beta^3\}),$$
where $\Delta$ is the maximum node degree.
Standard backtracking techniques can be used to compute a seed set of size at most $\beta$ that
influences this maximum number of nodes in the same $O(\min\{n\Delta^2\beta^3,n^2\beta^3\})$ time.
This concludes the proof of  Theorem \ref{theorem-tree}. \qed

\section{{The PES problem on Dense graphs}}\label{sec-dense}
In this section we concentrate on the PES problem in graphs characterized by large minimum degree.
In particular, we  relate the graph  minimum degree  to the size of the smallest  perfect seed set, e.g.,  a   set $S \subseteq V$ such that  $\I{S} = V$. 
\\
Assuming that $\ti(v)	\leq \ti $ and $\ta(v)\leq \ta$, for each $v \in V$,  and  $\ta+\ti \leq |V|+2$,
the algorithm PES($G,\ta,\ti$) selects and returns a set $S \subseteq V$, of size 
at most $2(\ta-1)$, that we will prove to be a PES for $G$ whenever the minimum degree of $G$ is $\frac{|V|+\ta+\ti}{2}-2$.

\begin{algorithm}
\SetCommentSty{footnotesize}
\SetKwInput{KwData}{Input}
\SetKwInput{KwResult}{Output}
\DontPrintSemicolon
\caption{ \ \   \textbf{Algorithm} PES($G,\ta,\ti$) \label{alg}}
\KwData { A graph $G=(V,E)$ having thresholds $\ti(v)\leq\ti$ and  $\ta(v)\leq \ta$ for $v\in V$. }
\KwResult{ $S,$ a perfect seed set for $G.$}
\setcounter{AlgoLine}{0}
Set $S$ as any subset of $V$ such that \\
\hspace{1truecm} -  $|S|=\ti$ and \\
\hspace{1truecm} - at least two nodes in $S$ are independent, if possible [e.g., if  $G$ is not a clique], \\
{\bf while} {($|S|< 2(\ta-1)$) \mbox{{\bf AND} ($\exists v\in V-S$ s.t. $|N(v) \cap S| \leq \ti-1$)}} {\bf do}
    		$S=S \cup \{v\}$\\
 \textbf{return} $S$
\end{algorithm}

The construction of the set $S$ returned by the 
algorithm PES($G,\ta,\ti$),
 immediately implies  the fact below.
\begin{fact}\label{fact}

1)  If $|S|< 2(\ta-1)$ then each $v \in V-S$ has at least $\ti$ neighbors in $S$.

2)  If $|S|= 2(\ta-1)$ then the sum of the degrees of the nodes in 
the subgraph   induced by $S$ in $G$
is upper bounded by
$$[\ti (\ti-1)-2] + 2(\ti -1)[2(\ta -1) -\ti]
= (\ti -1)(4\ta -\ti -4) -2$$ if $\ti \geq 2$; it is $0$ if $\ti=1$.
\end{fact}

\begin{theorem}
Let $G=(V,E)$ be a graph on $n$ nodes with $\ti(v) \leq \ti$, $\ta(v)\leq \ta$, for each $v\in V$,   where $\ta+\ti \leq n+2$,  and 
$d(v) \geq \frac{n+\ta+\ti}{2}-2$,  for each $v \in V$. 
The  algorithm \textsc{PES}($G,\ta,\ti$) returns a PES for $G$ of size at most 
$2\ta-2$.
\end{theorem}

\proof
Consider the  evangelization process in $G$ starting at the set $S$ returned by  the algorithm 
PES($G,\ta,\ti$). 
Let  
 $i\in\{0,1,\ldots\}$ be a round of the process and 
 $a(i)=|\A{S,i}-S|$ be the number of evangelists at round $i$ that  not belong to the seed set  $S$.
If $V-\I{S,i}=\emptyset$ then each node in $V-\A{S,i}$ has at least $\ti$ neighbors in $\A{S,i}$ and the theorem is proved.
Assume then   $V-\I{S,i}\neq \emptyset$. 
By 1) of Fact \ref{fact},  we know  that $|S| = 2(\ta-1)$.
Let  $\sigma(\A{S,i})$ denote the number of edges in the subgraph of $G$ induced by $\A{S,i}$.
In the following we assume that $\ti \geq 2$. 
The proof for $\ti=1$ can be obtained similarly recalling that the value in 2) of Fact \ref{fact} is $0$ in this case.
 By 2) of Fact \ref{fact} and since each node in $\A{S,i}-S$ is connected  at most to each other node in $\A{S,i}\cup S$, we have that sum of the degrees of the nodes in the subgraph of $G$ induced by $\A{S,i}$ is 
\begin{eqnarray}
2\sigma(\A{S,i}) &\leq &  
 (\ti -1)(4\ta {-}\ti {-}4) {-}2 + a(i)(a(i)-1) + 2a(i)[2(\ta-1)]  \\ \nonumber
&=& (\ti -1)(4\ta {-}4{-}\ti) {-}2 + a(i)^2 + a(i) (4\ta-5) 
\end{eqnarray}
Recalling  that $d(v) \geq \frac{n+\ta+\ti}{2}-2$ for each   $v \in V$, we get that 
the number \\$\sigma(\A{S,i},V-\A{S,i})$ of edges connecting one   node in  $\A{S,i}$ and one in  $V-\A{S,i}$ is

\begin{eqnarray} \nonumber
&&\sigma(\A{S,i},V-\A{S,i}) \geq \\ \nonumber 
\\ \nonumber
&& \quad \geq \frac{n{+}\ta{+}\ti{-}4}{2}[2(\ta{-}1){+}a(i)] 
                         - [(\ti -1)(4\ta {-}4{-}\ti) {-}2     + a(i)^2 +   a(i) (4\ta-5)] \\ \nonumber
										\\ \nonumber
&& \quad =(n+\ta+\ti-4)(\ta-1) - (\ti -1)(4\ta -4-\ti) +2  - a(i)^2 + \\ \nonumber
& &\quad \qquad\qquad\qquad\qquad\qquad \qquad\qquad\qquad \qquad\qquad\qquad 
         +  a(i)\left(\frac{n+\ta+\ti}{2} -4 \ta +3\right)\\ \nonumber
\\ \nonumber   
&& \quad = (n{+}\ta{-}3\ti)(\ta{-}1)+(\ti {-}1)\ti {+}2{-} a(i)^2 {+}   a(i)\left(\frac{n+\ta+\ti}{2} -4 \ta +3\right) \label{LB} \\
\end{eqnarray}
We first  determine   the minimum value of $a(i)$ that guaranties  that at least one   node $v\in V-\A{S,i}$ becomes an evangelist  at round $i+1$. 
By contradiction assume that each node in $V-\A{S,i}$ has at most $\ta-1$  neighbors in  $\A{S,i}$. This assumption implies that 
$\sigma(\A{S,i},V-\A{S,i})  \leq   (n-2(\ta -1) - a(i))(\ta-1)$

It is not hard to see that the lower bound  in (\ref{LB}) is larger than the above 
upper bound when 
 $0 \leq a(i) \leq \frac{n+\ta+\ti}{2} -3 \ta +2$.
This leads to a contradiction for such a  range of values of $a(i)$.
Hence, for each round $i$ for which $0 \leq a(i) \leq \frac{n+\ta+\ti}{2} -3 \ta +2$ at least one node $v\in V-\A{S,i}$ moves from $V-\A{S,i}$ to $\A{S,i+1}$ at round $i+1$.

\medskip
We  show now that if $a(i) = \frac{n+\ta+\ti}{2} -3 \ta +2$
(i.e.,  $|\A{S,i+1}| \geq 2(\ta-1) + \frac{n+\ta+\ti}{2} -3 \ta +3$)  then $|V-\I{S,i+1}|=0$, thus  completing the proof.\\
Indeed, we have
$|V-\A{S,i+1}| \leq n-[2(\ta-1) + \frac{n+\ta+\ti}{2} -3 \ta +3] = \frac{n-(\ta+\ti)}{2} + \ta -1$. This implies that 
 the number of evangelists among the  neighbors of any node  $v\in V-\A{S,i+1}$  is at least 

$$\frac{n+\ta+\ti}{2} -2 - \frac{n-(\ta+\ti)}{2} - \ta +2 = \ti.$$

\noindent
Hence, at round $i+1$ each node in $V-\A{S,i}$ is  influenced. Therefore,   $|V-\I{S,i+1}|=0$.
\qed
  
	\medskip
\noindent	We notice that in case $\ta=t_I=2$, we reobtain  the result  for Dirac graphs given in \cite{FPR}.
	\begin{corollary}
Let $G$ be a  graph with 
$d(v) \geq \frac{n}{2}$,  for each $v \in V$. 
The  algorithm PES($G,2,2$) returns an optimal PES  for $G$ of size
$2$.
\end{corollary}

\end{document}